\documentclass[11pt]{article}

\usepackage{amssymb,amsmath,amsthm,sectsty,url}
\usepackage[letterpaper,hmargin=1.0in,vmargin=1.0in]{geometry}
\usepackage[pdftex,colorlinks,linkcolor=blue,citecolor=blue,filecolor=blue,urlcolor=blue]{hyperref}
\usepackage{cleveref}
\usepackage{color}
\usepackage[boxed]{algorithm}
\usepackage{tikz}
\usepackage{graphicx}
\usepackage{nicefrac}
\usepackage{aliascnt}
\usepackage{caption}
\usepackage{subcaption}
\newtheorem{theorem}{Theorem}[section]
\crefname{theorem}{Theorem}{Theorems}

\newaliascnt{lemma}{theorem}
\newtheorem{lemma}[lemma]{Lemma}
\aliascntresetthe{lemma}
\crefname{lemma}{Lemma}{Lemmas}

\newaliascnt{proposition}{theorem}
\newtheorem{proposition}[proposition]{Proposition}
\aliascntresetthe{proposition}
\crefname{proposition}{Proposition}{Propositions}

\newaliascnt{corollary}{theorem}
\newtheorem{corollary}[corollary]{Corollary}
\aliascntresetthe{corollary}
\crefname{corollary}{Corollary}{Corollaries}

\newaliascnt{fact}{theorem}
\newtheorem{fact}[fact]{Fact}
\aliascntresetthe{fact}
\crefname{fact}{Fact}{Facts}

\newaliascnt{definition}{theorem}
\newtheorem{definition}[definition]{Definition}
\aliascntresetthe{definition}
\crefname{definition}{Definition}{Definitions}

\newaliascnt{remark}{theorem}
\newtheorem{remark}[remark]{Remark}
\aliascntresetthe{remark}
\crefname{remark}{Remark}{Remarks}

\newaliascnt{conjecture}{theorem}

\aliascntresetthe{conjecture}
\crefname{conjecture}{Conjecture}{Conjectures}

\newaliascnt{claim}{theorem}

\aliascntresetthe{claim}
\crefname{claim}{Claim}{Claims}

\newaliascnt{question}{theorem}

\aliascntresetthe{question}
\crefname{question}{Question}{Questions}

\newaliascnt{exercise}{theorem}

\aliascntresetthe{exercise}
\crefname{exercise}{Exercise}{Exercises}

\newaliascnt{example}{theorem}

\aliascntresetthe{example}
\crefname{example}{Example}{Examples}

\newaliascnt{notation}{theorem}

\aliascntresetthe{notation}
\crefname{notation}{Notation}{Notations}

\newaliascnt{problem}{theorem}

\aliascntresetthe{problem}
\crefname{problem}{Problem}{Problems}

\def\E{\mathbb E}

\newcommand{\N}{\mathbb N}
\newcommand{\R}{\mathbb R}

\begin{document}

\title{A Graph-Theoretic Approach to Multitasking}

\author{
  Noga Alon
  \thanks{Tel Aviv University,
  \texttt{nogaa@post.tau.ac.il}}
  \and
  Jonathan D. Cohen
  \thanks{Princeton University
  \texttt{jdc@princeton.edu}}
 \and
    Bisawdip Dey
\thanks{Princeton University
  \texttt{biswadip@princeton.edu}}
   \and
  Tom Griffiths
  \thanks{University of California, Berkeley
 \texttt{tom\_griffiths@berkeley.edu}}
  \and
  Sebastian Musslick
  \thanks{Princeton University
  \texttt{musslick@princeton.edu}}
  \and
  Kayhan {\"O}zcimder
  \thanks{Princeton University
  \texttt{ozcimder@exchange.princeton.edu}}
  \and
    Daniel Reichman
  \thanks{University of California, Berkeley
  \texttt{daniel.reichman@gmail.com}}
    \and
  Igor Shinkar
  \thanks{University of California, Berkeley
  \texttt{igor.shinkar@gmail.com}}
  \and
  Tal Wagner
 \thanks{MIT, CSAIL,
  \texttt{talw@mit.edu}}
  }

\maketitle

\begin{abstract}
 A key feature of neural network architectures is their ability to support the simultaneous
interaction among large numbers of units in the learning and
processing of representations. However, how the richness of such interactions
trades off against the ability of a network to simultaneously carry out multiple
independent processes -- a salient limitation in many domains of human cognition --
remains largely unexplored. In this paper we use a graph-theoretic analysis of network
architecture to address this question, where
tasks are represented as edges in a bipartite graph $G=(A \cup B, E)$.
We define a new measure of multitasking capacity of such networks, based on the assumptions that tasks that \emph{need} to be multitasked rely on independent resources, i.e., form a matching,
and that tasks \emph{can} be multitasked without interference if they form an induced matching.
Our main result is an inherent tradeoff between the multitasking capacity and the average degree of the network that holds \emph{regardless of the network architecture}.
These results are also extended to networks of depth greater than $2$. On the positive side, we demonstrate that networks that are random-like (e.g., locally sparse) can have desirable multitasking properties.
Our results shed light into the parallel-processing limitations of neural systems
and provide insights that may be useful for the analysis and design of parallel architectures.

\end{abstract}

\newpage
\tableofcontents
\newpage

\section{Introduction}

One of the primary features of neural network architectures is their ability to support parallel distributed processing. The decentralized nature of biological and artificial nets results in greater robustness and fault tolerance when compared to serial architectures such as Turing machines. On the other hand, the lack of a central coordination mechanism in neural networks can result in interference between units (neurons) and such interference effects have been demonstrated in several settings such as the analysis of associative memories \cite{amit1985storing} and multitask learning \cite{mccloskey1989catastrophic}. Understating the source of such interference and how it can be prevented has been a major focus of recent research (see, e.g., \cite{kirkpatrick2017overcoming} and the references therein).

Recently, a graph-theoretic model has suggested that interference effects may explain the limitations of the human cognitive system in multitasking: the ability to carry out multiple independent processes at the same time. This model consists of a simple 2-layer feed-forward network represented by a bipartite graph $G=(A \cup B,E)$ wherein the vertex set is partitioned into two disjoint sets of nodes $A$ and $B$, representing the inputs and the outputs of tasks respectively.
An edge $(a,b) \in E$ corresponds to a directed pathway from the input layer to the output layer in the network that is taken to represent a cognitive process (or task) that maps an input to an output \cite{neisser1967}. In more abstract terms, every vertex in $a \in A$ is associated with a set of inputs $I_a$, every vertex in $B$ is associated with a set of outputs $O_b$ and the edge $(a,b)$ is associated with a function $f_{a,b}:I_a\rightarrow O_b$ \footnote{The function $f_{a,b}$ is hypothesized to be implemented by a gate used in neural networks such as sigmoid or threshold gate.}.
In this work, we also consider deeper architectures with $r>2$ layers, where edges correspond to mappings between nodes from consecutive layers and a path $P$ from the input (first) layer to the output (last) layer is simply the composition of the mappings on the edges in $P$. The model above is quite general and simple modifications of it may apply to other settings. For example, we can assume the vertices in $A$ are senders and vertices in $B$ are receivers and that a task associated with an edge $e=(a,b)$ is transmitting information from $a$ to $b$ along a communication channel $e$.

Given a 2-layer network, a \emph{task set} is a set of edges $T \subseteq E$. A key assumption made in \cite{feng2014multitasking} that we adopt as well is that all task sets that need to be multitasked in parallel form a \emph{matching}, namely, no two edges in $T$ share a vertex as an endpoint. This assumption reflects a limitation on the parallelism of the network that is similar to the Exclusive Read Exclusive Write (EREW) model in parallel RAM, where the tasks cannot simultaneously read from the same input or write to the same output. Similarly, for depth $r>2$ networks, task sets correspond to \emph{node disjoint} paths from the input layer to the output layer. For simplicity, we shall focus from now on the depth 2 case with $|A|=|B|=n$.

In \cite{Musslick2016a,feng2014multitasking} it is suggested that concurrently executing two tasks associated with two (disjoint) edges $e$ and $f$ will result in interference if $e$ and $f$ are connected by a third edge $h$. The rationale for this interference assumption stems from the distributed operation of the network that may result in the task associated with $h$ becoming activated automatically once its input and output are operating, resulting with interference with the tasks associated with $e$ and $f$. Therefore, \cite{Musslick2016a,feng2014multitasking} postulate that all tasks within a task set $T$ can be performed in parallel without interferences only if the edges in $T$ form an \emph{induced} matching. Namely, no two edges in $T$ are connected by a third edge. Interestingly, the induced matching condition also arises in the communication setting \cite{birk1993uniform,alon2012nearly,chlamtac1985broadcasting}, where it is assumed that messages between senders and receivers can be reliably transmitted if the edge set connecting these nodes forms an induced matching. Following the aforementioned interference model, \cite{Musslick2016a,feng2014multitasking} define the multitasking capability of a bipartite network $G$ as the maximum cardinality of an induced matching in $G$.

The main message of \cite{Musslick2016a,feng2014multitasking} is that there is a fundamental tradeoff in
neural network architectures like the human brain between the \emph{efficiency} of shared representations
, and the independence of representations that supports concurrent multitasking (this tradeoff is termed ``multitasking versus multiplexing''). In graph-theoretic terms, it is suggested that as the average degree $d$ (``efficiency of representations''--larger degree corresponds to more economical and efficient use of shared respresentations)
of $G$ increases, the ``multitasking ability'' should decay in $d$. In other words, the cardinality of the maximal induced matching should be upper bounded by $f(d)n$ with $\lim_{d \rightarrow \infty} f(d)=0$. This prediction was tested and supported on certain architectures by numerical simulations in \cite{Musslick2016a,feng2014multitasking}.
Establishing such as a tradeoff is of interest, as it can identify limitations of artificial nets that rely on shared representations and aid in designing systems that attain an optimal tradeoff. Furthermore, such a tradeoff is also of significance for cognitive neuroscience as it can shed some light on the source of the striking limitation of the human cognitive system to execute control demanding tasks simultaneously.


Identifying the multitasking capacity of $G=(A \cup B,E)$ with the size of its maximal induced matching has two drawbacks.
First, the fact that there is some, possibly large, set of tasks that can be multitasked does not preclude the existence of a (possibly small) set of critical tasks that greatly interfere with each other (e.g., consider the case in which a complete bipartite graph $K_{d,d}$ occurs as a subgraph of $G$. This is illustrated in Figure~\ref{fig:multitask_example}). Second, it is easy to give examples of graphs (where $|A|=|B|=n$) with arbitrarily large average degree that nonetheless contain an induced matching of size $n/2$.
For example, there are $d$-regular bipartite graphs with $n$ vertices on each side that contain an induced matching
of size $n/2$ even when $d=\Omega(n)$ (For example, one can take two copies of a dense bipartite graph $F$ and connect these two copies with a perfect matching-see Figure~\ref{fig:multitask_example} for an illustration).
Hence, it is impossible to upper bound the multitasking capacity of every network with average degree $d$ by $f(d)n$ with $f$ vanishing as the average degree $d$ tends infinity.
Therefore, the generality of the suggested tradeoff between efficiency and concurrency is not clear under this definition.
\begin{figure}
    \centering
    \includegraphics[scale=0.6]{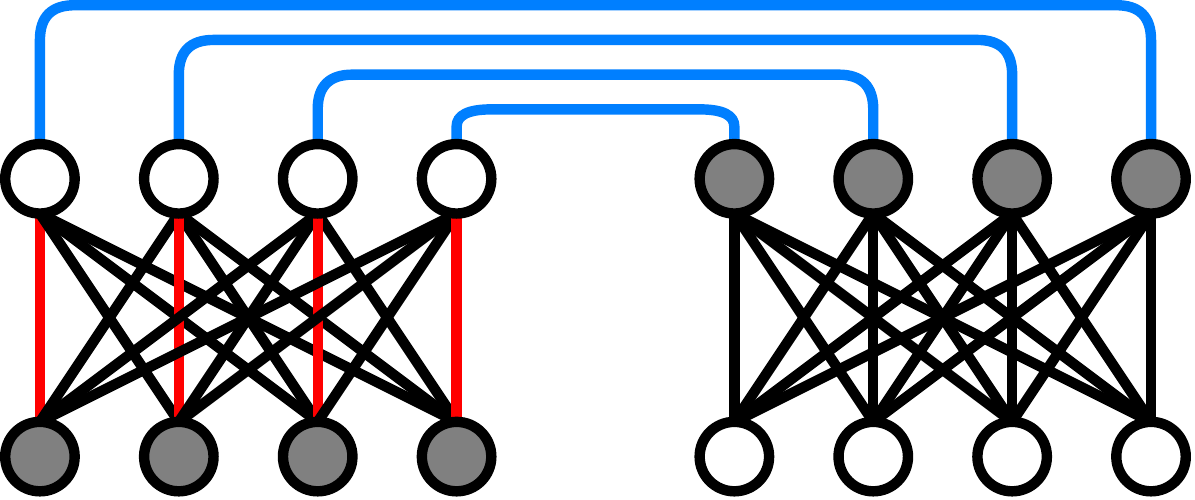}
    \caption{In the depicted bipartite graph, the node shading represents the bipartition. The blue edges form an induced matching, which represents a large set of tasks that can be multitasked. However, the red edges form a matching in which the largest induced matching has size only $1$. This represents a set of tasks that greatly interfere with each other.}
    \label{fig:multitask_example}
\end{figure}

Our main contribution is a novel measure of the multitasking capacity that is aimed at solving the first problem, namely networks with ``high'' capacity that contain a task set whose edges badly interfere with one another.
In particular, for a parameter $k$ we consider \emph{every} matching of size $k$,
and ask whether every matching $M$ of size $k$ contains a large \emph{induced} matching $M'\subseteq M$. This motivates the following definition (see Figure~\ref{fig:hypercube} for an illustration).

\begin{definition}
Let $G=(A \cup B,E)$ be a bipartite graph with $|A|=|B|=n$, and let $k \in \N, k \leq n$ be a parameter.
We say that $G$ is a $(k,\alpha(k))$-multitasker
if for every matching $M$ in $G$ of size $|M| = k$,
there exists an induced matching $M' \subseteq M$ such that
\[
    |M'| \geq \alpha(k) |M|.
 \]
We will say that a graph $G$ is an $\alpha$-multitasker if it is $(k,\alpha)$-multitasker for all $k=1,\dots,n$.

The parameter $\alpha \in (0,1]$ measures the multitasking capabilities of $G$,
and the larger $\alpha$ is the better multitasker $G$ is considered.
We call the parameter $\alpha(k) \in (0,1]$ the \emph{multitasking capacity} of $G$ for matchings of size $k$.
\end{definition}

Our definition generalizes without much difficulty to networks of depth $r>2$, where instead of matchings, we consider first to last
node disjoint paths, and instead of induced matchings we consider induced paths, i.e., a set of disjoint paths such that no two nodes belonging to different paths are adjacent.

Observe that our measure is related to the previously mentioned measure of the cardinality of an induced matching.
That is, if $G$ is an $(n,\alpha(n)$-multitasker for a large $\alpha(n)$, then $G$ contains a large induced matching.

\begin{figure}
    \centering
    \begin{subfigure}[b]{0.3\textwidth}
        \includegraphics[width=0.7\textwidth]{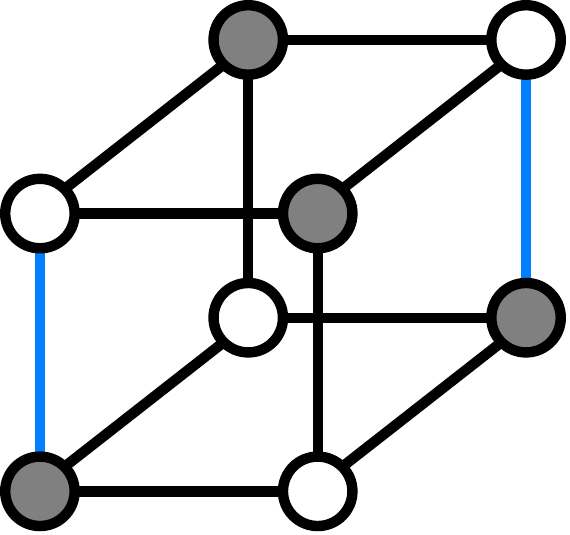}
    \end{subfigure}
 	\quad
    \begin{subfigure}[b]{0.3\textwidth}
        \includegraphics[width=0.7\textwidth]{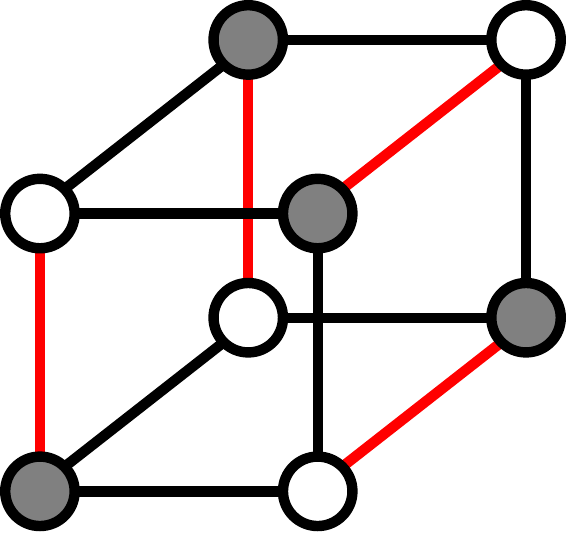}
    \end{subfigure}
    \caption{The hypercube on $8$ nodes. The node shading represents the bipartition. On the left, the blue edges form an induced matching of size $2$. On the right, the red edges form a matching of size $4$ whose largest induced matching has size $1$, and hence the multitasking capacity of the hypercube is at most $1/4$.}
    \label{fig:hypercube}
\end{figure}

The main question we shall consider here is what kind of tradeoffs one should expect between $\alpha, d$ and $k$. In particular, are there networks with
large average degree that achieve a multitasking capacity bounded away from $0$, especially, if $k$ is not too large? Which network architectures give rise to good multitasking behavior? Should we expect ``multitasking vs. multiplexing'': namely, $\alpha(k)$ tending to zero with $d$ for all graphs of average degree $d$? While our definition of multitasking capacity is aimed at resolving the problem of small task sets that can be poorly multitasked, it turns out to be also related also to the ``multitasking vs. multiplexing'' phenomena. Furthermore, our graph-theoretic formalism also gives insights as to how network depth and interferences are related.

\subsection{Our results}
We provide some answers to the questions raised above.
Our main contribution is in establishing a \emph{tradeoff} between multitasking capacity of a graph
and the its edge density that hold for \emph{arbitrary} networks.

We divide the presentation of the results into two parts.
The first part discusses the case of $d$-regular graphs,
and the second part discusses general graphs.
\paragraph{The $d$-regular case:}
Let $G = (A \cup B, E)$ be a bipartite $d$-regular graph with $n$ vertices on each side.
Considering the case of $k=n$, i.e., maximal possible induced matchings that are contained in a \emph{perfect matching},
we show that if a $d$-regular graphs is an $(n,\alpha(n))$-multitasker, then $\alpha(n)=O(1/\sqrt{d})$.
Our upper bound on $\alpha(n)$ establishes an inherent limitation on the multitasking capacity of any network.
That is, for any task set of size $n$ it holds that $\alpha(n)$ must tend to $0$ as the degree grows.
In fact, we prove that degree of the graph $d$ constrains the multitasking capacity also for task sets of smaller sizes.
Specifically, for $k$ that is sufficiently larger than $\Omega(n/d)$ it holds that $\alpha(k)$ tends to $0$ as $d$ increases.
We summarize these results in the following theorem.
\begin{theorem}\label{thm:d reg upper bound}
There is a constant $\gamma \in \R_+$ such that the following holds.
Let $G=(A \cup B,E)$, be a $d$-regular bipartite graph with $|A|=|B|=n$.
\begin{enumerate}
\item If $n/d^{1/4} \leq k  \leq n$, then $\alpha(k) \leq O(\frac{n}{k\sqrt{d}})$.
In particular, there exists a perfect matching in $G$ that does not contain an induced matching of size larger than $O(n/\sqrt{d})$.
\item If $n/d^{1/3} \leq k \leq n/d^{1/4}$, then $\alpha(k) \leq O(k/n)$.
\item If $\gamma n/d \leq k \leq n/d^{1/3}$ then $\alpha(k) \leq O(\sqrt{\frac{n}{kd}})$.
\end{enumerate}
\end{theorem}

For a certain range of parameters our results are tight.
Specifically, when considering task sets of size $n$ our result is tight up to logarithmic factors,
as we provide a construction of a $d$-regular graph where every matching of size $n$ contains an induced matching of size
$\Omega(\frac{1}{\sqrt{d \log d}})$. See \Cref{thm:d reg alpha = 1/sqrt(d log d)} for details.

For arbitrary values of $k \leq n$ it is not hard to see that every $d$-regular graph achieves $\alpha(k)\geq\frac1{2d}$.
We show that this naive bound can be asymptotically improved upon, by constructing
an $\alpha$-multitaskers with $\alpha=\Omega(\frac{\log d}{d})$.
The construction is based on bipartite graphs which have good spectral expansion properties.
See \Cref{thm:expander} for details.

Considering bounded values of $k$ we show that is it possible
to achieve multitasking capacity bounded away above $0$,
when measured on task sets of bounded size (up to $k$).
The best multitasking capacity one can hope for is $\alpha=1/2$ (see Remark~\ref{rem:alpha_half}),
and we construct $(k,1/2)$-multitaskers for all $k \leq O(\log_d(n))$.
See \Cref{thm:large girth} for details.

We also consider networks of depth $r>2$ \footnote{We think of $r$ as a constant independent of $n$ and $d$ as tending to infinity with $n$.}. We generalize our ideas for depth 2 networks by upperbounding the multitasking capacity of arbitrary $d$-regular networks of depth $r$ by $O(\frac{r}{d^{1-1/r}})$. In particular, we show that such networks must contain a family $S$ of paths of size $n$ such that every set of induced paths contained in $S$ has size at most $O(\frac{r}{d^{1-1/r}})n$. Observe that this shows that for tasks sets of size $n$, network of depth $2 < r \ll d$ incur interference which is strictly worse than depth $2$ networks. We believe that it is also the case that interference gets worst with $r$
(namely that interference worsens as $r$ increases to $r+1$ for $r>2$), although whether this is indeed the case is an open problem.

\paragraph{The irregular case:}

Next we turn to arbitrary, not necessarily regular, graphs.
We show that for an arbitrary bipartite graph with $n$ vertices on each side and \emph{average} degree $d$ its multitasking capacity
$\alpha(n)$ is upper bounded by $O\left(\frac{\log n}{d}\right)^{1/3}$.
That is, when the average degree is concerned,
the multitasking capacity of a graph tends to zero,
provided that the average degree of a graph is larger than $\log(n)$.

\begin{theorem}\label{thm:upper bound avgdeg}
There is a constant $\gamma \in \R_+$ such that the following holds.
Let $G=(A \cup B,E)$, be a bipartite graph of average degree $d$ with $|A|=|B|=n$.
If $G$ is an $\alpha$-multitasker then $\alpha \leq O((\frac{\log n}{d})^{1/3}))$.

\end{theorem}

For dense graphs satisfying $d=\Omega(n)$ (which are studied in \cite{feng2014multitasking}),
we prove a stronger upper bound of $\alpha(n)=O(\frac{1}{\sqrt{n}})$ using the well known Szemer\'{e}di regularity lemma.
See \Cref{t44} for details.

We also show that there are multitaskers of average degree $\Omega(\log \log n)$, with $\alpha>1/3-\epsilon$.
Hence, in contrast to the regular case, for the multitasking capacity to decay with \emph{average} degree $d$,
we must assume that $d$ grows faster than $\log \log n$.
See \Cref{t51} and \Cref{t52} for the exact statements.
It is an interesting question whether there exists a multitasker with $\alpha>0$ independent of $n$, for average degree $\Theta(\log n)$,
which, if true is the largest average degree possible. This is left as an open problem.

Finally, for any $d \in \N$ we show a construction of a graph $G$ with average degree $d$
such that for every $0<\alpha<1/5$, $G$ is a $(k,\alpha)$-multitaskers for all $k \leq \Omega(n/d^{1+4\alpha})$.
Comparing this to the foregoing results, here we do not required that $d=O(\log \log n)$.
Allowing larger values of $d$ allows for weaker multitasking: we obtain that the graph is a multitasker only with
respect to matchings whose size is at most $n/d^{1+4\alpha}$.
See \Cref{thm:G n d over n} for details.

\section{Preliminaries}

A matching $M$ in a graph $G$ is a set of edges $\{e_1,...,e_m\}$ such that no two edges in $M$ share a common vertex. If $G$ has $2n$ vertices and $|M|=n$, we say that $M$ is a perfect matching. By Hall Theorem, every $d$-regular graph with bipartition $(A,B)$ has a perfect matching. A matching $M$ is \emph{induced} if there are no two distinct edges $e_1,e_2$ in $M$, such that there is an edge connecting $e_1$ to $e_2$.
Given a graph $G=(V,E)$ and two disjoint sets $A,B\subseteq V$ we let $e(A,B)$ be the set of edges with one endpoint in $A$ and the other in $B$. For a subset $A$, $e(A)$ is the set of all edges contained in $A$. Given an edge $e \in E$, we define the graph $G/e$ obtained by contracting $e=(u,v)$ as the graph with a vertex set $(V\cup v_{e}) \setminus \{u,v\}$. The vertex $v_{e}$ is connected to all vertices in $G$ neighboring $u$ or $v$. For all other vertices $x,y \in V \setminus\{u,v\}$, they form an edge in
$G/e$ if and only if they were connected in $G$. Contracting a set of edges, and in particular contracting a matching, means contracting the edges one by one in an arbitrary order.

Given a subset of vertices $U\subseteq V$, the subgraph induced by $U$, denoted by $G[U]$ is the graph whose vertex set is $U$ and two vertices in $U$ are connected if and only if they are connected in $G$.
For a set of edges $E' \subseteq E$, denote by $G[E']$ the graph induced by all vertices incident to an edge in $E'$.
We will use the following simple observation throughout the paper.

\begin{lemma}\label{lem:contraction}
Let $M$ be a matching in $G$, and let $d_{avg}$ be the average degree of $G[M]$.
Suppose that we contract all edges in $M$ in $G[M]$. Then the resulting graph
$\widetilde{G}[M]$ has average degree at most $2d_{avg}-2$.
\end{lemma}

\begin{proof}
$G[M]$ contains $2|M|$ vertices and $d_{avg}|M|$ edges. The result follows as $\widetilde{G}[M]$ has $|M|$ vertices and at most $d_{avg}|M|-|M|$ edges.
\end{proof}

An \emph{independent set} in a graph $G=(V,E)$ is a set of vertices that do not span an edge.
We will use the following well known fact attributed to Turan.

\begin{lemma}\label{lem:turan_independent_set}
Every $n$-vertex graph with average degree $d_{avg}$ contains an independent set of size at least $\frac{n}{d_{avg}+1}$.
\end{lemma}
The \emph{girth} of a graph $G$ is the length of the shortest cycle in $G$.

Let $G=(V,E)$ be a bipartite graph, $k$ an integer and $\alpha \in (0,1]$, a parameter. We define the $(\alpha,k)$-matching graph $H(G,\alpha,k)=(L,R,F)$ to
be a bipartite graph where $L$ is the set of all matchings of size $k$ in $G$, $R$ is the set of all induced matchings of size $\alpha k$ in $G$ and a vertex $v_M \in L$
(corresponding to matching $M$ of size $k$) is connected to a vertex $u_{M'}$ (corresponding to an induced matching $M'$ of size $\alpha k$) if and only if
$M' \subseteq M$. We omit $\alpha,k,G$ from the notation of $H$ when it will be clear from the context.
We will repeatedly use the following simple Lemma in upper bounding the multitasking capacity in graph families.
We refer to this Lemma as the induced matching Lemma.
\begin{lemma}\label{lem:induced_matching}
Suppose the average degree of a vertex in $L$ in the graph $H(G,\alpha,k)$ is strictly smaller than $1$.
If $G$ is a $(k,\alpha(k))$-multitasker, then $\alpha(k)< \alpha$.
\end{lemma}

\begin{proof}
By the assumption, $L$ has a vertex of degree 0. Hence there exist a matching of size $k$ in $G$ not containing an induced matching of size $\alpha k$. As required.
\end{proof}

Throughout the paper we will need the following concentration inequalities known as Chernoff's bound.
\begin{lemma}\label{lmm:chernoff}
Let $X_1 \ldots X_n$ be $\{0,1\}$ independent random variables where for every $\Pr[X_i=1]=p$ for all $i=1,\dots, n$,
and let $X=\sum_{i=1}^n X_i$. Then, for all $\eta \in (0,1)$ it holds that
\begin{equation*}
    \Pr[X<(1-\eta)pn] < \exp(-\frac{\eta^2 p n}{2})
\end{equation*}
and
\begin{equation*}
    \Pr[X>(1+\eta)pn] < \exp(- \frac{\eta^2 pn}{2+\eta}) \enspace.
\end{equation*}
\end{lemma}

\section{Upper bounds on the multitasking capacity}\label{sec:upper bounds}

\subsection{The regular case}
In this section we prove Theorem~\ref{thm:d reg upper bound}
that upper bounds the multitasking capacity
of arbitrary $d$-regular multitaskers.
We start the proof of Theorem~\ref{thm:d reg upper bound}
with the case $k=n$. The following theorem shows that $d$-regular $(k=n,\alpha)$-multitaskers must have $\alpha=O(1/\sqrt d)$.

\begin{theorem}\label{thm:upper_bound_n}
Let $G=(A \cup B,E)$, be a bipartite $d$-regular graph where $|A|=|B|=n$.
Then $G$ contains a perfect matching $M$ such that every induced matching
$M' \subseteq M$ has size at most $\frac{9 n}{\sqrt d}$.
\end{theorem}

For the proof, we need the following bounds on the number
of perfect matchings in $d$-regular bipartite graphs.
\begin{lemma}\label{lmm:num_perfect_matchings}
Let $G=(A,B,E)$, be a bipartite $d$-regular graph where $|A|=|B|=n$. Denote by $M(G)$ the number of perfect matchings in $G$. Then
\[
    \left(\frac{d}{e}\right)^n\le \left(\frac{(d-1)^{d-1}}{d^{d-2}}\right)^n\le M(G) \le (d!)^{n/d}.
\]
\end{lemma}
The lower bound on $M(G)$ is due to Schrijver~\cite{schrijver1998counting}.
The upper bound on $M(G)$ is known as Minc's conjecture, which has been proven by Bregman~\cite{bregman1973some}.

\begin{proof}[Proof of Theorem~\ref{thm:upper_bound_n}]

Consider $H(G,\alpha,n)$, where $\alpha$ will be determined later.
Clearly $|R| \leq {n \choose \alpha n}^2 \leq (\frac{e}{\alpha})^{2\alpha n}$.
By the upper bound in Lemma~\ref{lmm:num_perfect_matchings},
every induced matching of size $\alpha n$ can be contained in at most $(d!)^{(1-\alpha)n/d}$ perfect matchings.
By the lower bound in Lemma~\ref{lmm:num_perfect_matchings},
$|L| \geq \left(\frac{d}{e}\right)^n$.
Therefore, the average degree of the the vertices in $L$ is at most
\[
  \frac{(\frac{e}{\alpha})^{2\alpha n} \cdot (d!)^{(1-\alpha)n/d}}{\left(\frac{d}{e}\right)^n} \leq
  \frac{(\frac{e}{\alpha})^{2\alpha n} \cdot (\sqrt{2\pi d}(\tfrac{d}{e})^d)^{(1-\alpha)n/d}}{\left(\frac{d}{e}\right)^n}
  = \left( \frac{e^3}{\alpha^2d} \cdot (2 \pi d)^{\frac{1-\alpha}{2 \alpha d}} \right)^{\alpha n}.
\]
Setting $\alpha > 2\sqrt{\frac{e^3}{d}}$ yields $\frac{e^3}{\alpha^2d}<\frac{1}{2}$,
and it can be verified that $(2 \pi d)^{\frac{1-\alpha}{2 \alpha d}} < 2$ for all such $\alpha$.
Therefore in this setting, the average degree of the vertices in $L$ is smaller than $1$, which
concludes the proof by Lemma~\ref{lem:induced_matching}.
This completes the proof of the theorem.
\end{proof}

We record the following simple observation, which is immediate from the definition.
\begin{proposition}\label{prp:mtasker_pertrub}
If $G$ is a $(k,\alpha)$-multitasker, then for all $1 < \beta \leq n/k$,
the graph $G$ is a $(\beta k,\frac{\alpha}{\beta})$-multitasker.
\end{proposition}
By combining Theorem~\ref{thm:upper_bound_n} with (the contrapositive of) Proposition~\ref{prp:mtasker_pertrub} we obtain the following immediate corollary.
\begin{corollary}\label{cor:upper_bound_high_range}
If $G$ is a $d$-regular $(k,\alpha)$-multitaskers with $n$ vertices on each side
and $k > n/\sqrt{d}$, then $\alpha \leq O(\frac{n}{k\sqrt{d}})$.
\end{corollary}

Next, we prove that for smaller values of $k$ the multitasking capacity $\alpha(k)$
is upper bounded by $O \left( \max(\frac{n}{k}, \sqrt\frac{kd}{n}) \right)$.

\begin{theorem}\label{thm:upper_bound_k}
Let $G=(A \cup B,E)$ be a $d$-regular (bipartite) subgraph with $|A|=|B|=n$, and let $k < n/2$.
Then,$G$ contains a matching $M$ of size $k$, such that every induced matching
$M' \subset M$ has size $|M'| \leq \alpha k$ for $\alpha = \max\{\frac{4k}{n} , 9\sqrt\frac{n}{kd}\}$.
\end{theorem}

In particular, this rules out the existence $d$-regular $(\omega(n/d),\alpha)$-multitaskers
for any constant multitasking capacity $\alpha>0$. To see this, take any $\epsilon>0$, and put
$d\geq\frac{4\cdot81}{\epsilon^3}$ and
$k=\frac{81}{\epsilon^2}\cdot\frac{n}{d}$ in the above theorem.
It implies that $\alpha\leq\epsilon$.


In the proof of Theorem~\ref{thm:upper_bound_k} we use the following result on the number of matchings of size $k$ in $d$-regular bipartite graphs, known as the~\emph{Lower Matching Conjecture} and recently proven by Csikv\`{a}ri~\cite{csikvari2014lower}.

\begin{lemma}\label{lmm:num_k_matchings}
Let $G=(A \cup B,E)$, be a bipartite $d$-regular graph where $|A|=|B|=n$. Denote by $M_k(G)$ the number of matchings of size $k$ in $G$. Then
\[
M_k(G) \geq {n\choose k}^2\left(1-\frac{k}{nd}\right)^{nd-k} \left(\frac{kd}{n}\right)^k .
\]
\end{lemma}

In Appendix~\ref{section:num_k_matching} we derive from Lemma~\ref{lmm:num_k_matchings} the following bound.
\begin{corollary}\label{eq:num_k_matchings}
In the setting of Lemma~\ref{lmm:num_k_matchings}, if $k < n/2$, then
\[ M_k(G) \geq \left(\frac{end}{k}\right)^k \cdot \left(\frac{1}{2e}\right)^{4k^2/n} \cdot \frac{1}{2\pi k} . \]
\end{corollary}

\begin{proof}[Proof of Theorem~\ref{thm:upper_bound_k}]
For brevity, we refer to a matching of size $k$ as a~\emph{$k$-matching}.
Fix $\alpha\in(0,1]$.
Consider the graph $H=(G,\alpha,k)$.
Clearly $|R| \leq {n \choose \alpha k}^2 \leq (\frac{en}{\alpha k})^{2\alpha k}$.
For a given induced $\alpha k$-matching, we can obviously upper-bound the number of $k$-matchings that contain it by the total number of edge subsets of size $k$ that contain it, which is at most ${nd \choose k-\alpha k} \leq \left(\frac{end}{(1-\alpha)k}\right)^{(1-\alpha)k}$.
By Corollary~\ref{eq:num_k_matchings},
$|L| \geq \frac{1}{2\pi k}\left(\frac{end}{k}\right)^k \left(\frac{1}{2e}\right)^{4k^2/n}$.
Therefore, the average degree of the vertices in $L$ is at most
\[
  \frac{(\frac{en}{\alpha k})^{2\alpha k} \cdot \left(\frac{end}{(1-\alpha)k}\right)^{(1-\alpha)k}}{\frac{1}{2\pi k}\left(\frac{end}{k}\right)^k \left(\frac{1}{2e}\right)^{4k^2/n}}
  =
  \left( \frac{en}{\alpha^2kd} \right)^{\alpha k}
  \cdot \left(\frac{1}{1-\alpha}\right)^{(1-\alpha)k}
  \cdot (2e)^{4k^2/n}
   \cdot 2\pi k  .
\]
If we choose $\alpha$ such that this bound is smaller than $1$, then there must be a vertex in $L$ with no neighbors in $R$, and we are done by Lemma~\ref{lem:induced_matching}. Hence we need $\alpha$ to satisfy
\[
  \frac{en}{\alpha^2kd}
  \cdot \left(\frac{1}{1-\alpha}\right)^{(1-\alpha)/\alpha}
  \cdot (2e)^{4k/(\alpha n)}
   \cdot (2\pi k)^{1/(\alpha k)} < 1 .
\]
We now bound the terms on the left-hand side for an appropriate choice of $\alpha$.
For $\alpha>4/\sqrt{k}$  the term $(2\pi k)^{1/(\alpha k)}$ is upper bounded by 2.
For $\alpha>4k/n$ the term $(2e)^{4k/(\alpha n)}$ is upper bounded by $2e$.
The term $\left(1-\alpha\right)^{-(1-\alpha)/\alpha}$ is upper-bounded by $e$ for any $\alpha$.
Therefore, if we chose $\alpha$ that satisfies both of the above inequalities,
the average degree of the vertices in $L$ is at most
$\frac{4e^3n}{\alpha^2kd} < 81 \frac{n}{\alpha^2kd}$, which is smaller than 1 for $\alpha > 9\sqrt{\frac{n}{kd}}$.
Overall, choosing $\alpha > \max\{4/\sqrt{k} , \frac{4k}{n} , 9\sqrt\frac{n}{kd}\}$ suffices. By noting that $9\sqrt{\frac{n}{kd}} > 4/\sqrt{k}$, we get $\alpha > \max\{\frac{4k}{n} , 9\sqrt\frac{n}{kd}\}$ as stated.
\end{proof}

\paragraph{Putting all the bounds together}
Note that the bound $\max\{\frac{4k}{n} , 9\sqrt\frac{n}{kd}\}$
from Theorem~\ref{thm:upper_bound_k} is equal to $\frac{4k}{n}$ if $k > \frac{n}{(16d/81)^{1/3}}$, and $9\sqrt\frac{n}{kd}$ otherwise.
Combining this with the bound in Corollary~\ref{cor:upper_bound_high_range}
we obtain Theorem~\ref{thm:d reg upper bound}.

\subsection{Upper bounds for networks of depth larger than $2$}
A graph $G(V,E)$ is a~\emph{network} with $r$ layers of width $n$
and degree $d$, if $V$ is partitioned into $r$ independent sets
$V_1,\ldots,V_r$ of size $n$ each, such that each $(V_i,V_{i+1})$ induced
a $d$-regular bipartite graph for all $i<r$, and there are no additional
edges in $G$.

A~\emph{top-bottom} path in $G$ is a path $v_1,\ldots,v_r$ such that
$v_i\in V_i$ for all $i\leq r$, and $v_i,v_{i+1}$ are neighbors for
all $i<r$.

A set of node-disjoint top-bottom paths $p_1,\ldots,p_k$ is
called~\emph{induced} if for every two edges $e\in p_i$ and $e'\in
p_j$ such that $i\neq j$, there is no edge in $G$ connecting $e$
and $e'$.

\begin{fact}
    A set of node-disjoint top-bottom paths $p_1,\ldots,p_k$ is induced
    if and only if for every $i<r$
    it holds that $(p_1\cup \ldots \cup p_k) \cap E(V_i,V_{i+1})$ is an induced matching in $G$.
\end{fact}

We say that a network $G$ as above is a $(k,\alpha)$-multitasker if every
set of $k$ node-disjoint top-bottom paths contains an induced subset of
size at least $\alpha k$.

\begin{theorem}
If $G$ is an $(n,\alpha)$-multitasker
then $\alpha < e(er/d)^{1-\frac1r}=O(r/d)^{1-1/r}$.
\end{theorem}
\begin{proof}
Let $H(L,R;E_H)$ be the bipartite graph in which side $L$ has a node for
each set of $n$ node-disjoint top-bottom paths in $G$, side $R$ has a
node for each induced set of $\alpha n$ node-disjoint top-bottom paths
in $G$, and $P\in L$, $P'\in R$ are adjacent iff $P'\subset P$. Let $D$
be the maximum degree of side $R$. We wish to upper-bound the average
degree of side $L$, which is upper-bounded by $D|R|/|L|$.

$|R|$ is clearly upper bounded by ${n\choose\alpha n}^r$. It is a
simple observation that $|L|$ equals $\prod_{i<r}m_i$, where $m_i$
denotes the number of perfect matchings in the bipartite graph $G[V_i\cup
V_{i+1}]$. Since this graph is $d$-regular, by
the Falikman-Egorichev proof of the Van der Waerden conjecture
(\cite{falikman1981proof}, \cite{egorychev1981solution}),
or by Schrijver's lower bound, we
have $m_i\geq(d/e)^n$ and hence $|L|\geq(d/e)^{n(r-1)}$. To upper bound
$D$, fix $P'\in R$, and let $G'$ be the network resulting by removing
all nodes and edges in $P'$ from $G$. This removes exactly $\alpha n$
nodes from each layer $V_i$; denote by $V_i'$ the remaining nodes in
this layer in $G'$. It is a straightforward observation that $D$ equals
the number of sets of $(1-\alpha)n$ node-disjoint top-bottom paths in
$G'$. Each such set decomposes into $M_1,\ldots,M_{r-1}$ such that $M_i$
is a perfect matching on $G'[V_i',V_{i+1}']$ for each $i<r$. Therefore
$D\leq\prod_{i-1}m_i'$ where $m_i'$ denotes the number of perfect
matchings in $G'[V_i',V_{i+1}']$. The latter is a bipartite graph with
$(1-\alpha)n$ nodes on each side and maximum degree $d$, and hence by the
Bregman-Minc inequality, $m_i'\leq(d!)^{(1-\alpha)n/d}$. Consequently,
$D\leq(d!)^{(1-\alpha)n(r-1)/d}$.

Putting everything together, we find
that the average degree of side $L$ is upper bounded by
\[
  \frac{D|R|}{|L|} \leq
  \frac{(d!)^{(1-\alpha)n(r-1)/d}
\cdot {n\choose\alpha n}^r}{(d/e)^{n(r-1)}} \leq
  \frac{(\sqrt{2\pi d}(d/e)^d)^{(1-\alpha)n(r-1)/d}
\cdot (\frac{e}{\alpha})^{\alpha nr}}{(d/e)^{n(r-1)}}
\]
\begin{equation}\label{eq:layers1}
  = \left((2\pi d)^{\frac{1-\alpha}{2\alpha d}}\cdot\frac{e}{d}
\left(\frac{e}\alpha\right)^{\frac{r}{r-1}}\right)^{\alpha n(r-1)}.
\end{equation}
We will show that if $\alpha \geq e(er/d)^{1-\frac1r}$ then above bound
is less than $1$, which implies side $L$ has a node of degree $0$, a
contradiction. To this end, note that for this setting of $\alpha$ we have
\begin{equation}\label{eq:layers2}
  \frac{e}{d}\left(\frac{e}\alpha\right)^{\frac{r}{r-1}} \leq \frac{1}{r},
\end{equation}
and
\[ (2\pi d)^{(1-\alpha)/(2\alpha d)} \leq
(2\pi d)^{1/(2\alpha d)} \leq (2\pi d)^{1/(2e(er)^{1-1/r}d^{1/r})} . \]
One can verify that,
\begin{fact}
For every constants $\alpha,\beta>0$, the function
$f(d)=(\alpha d)^{1/(\beta d^{1/r})}$ is maximized at $d=e^r/\alpha$.
\end{fact}
Plugging this above (and using $r\geq 2$), we obtain
\[
  (2\pi d)^{(1-\alpha)/(2\alpha d)} \leq
  e^{(2\pi r)^{1/r}/(2e^2)} \leq
  e^{\sqrt{2\pi}\cdot e^{1/e}/(2e^2)} < 1.28 < r,
\]
and plugging this with~\cref{eq:layers2} into~\cref{eq:layers1}
yields $\frac{D|R|}{|L|}<1$, as needed.
\end{proof}

\subsection{The irregular case}

Below we consider general graphs with average degree $d$. This is in contrast
to the previous section, where we considered only $d$-regular graphs.

\begin{theorem}
\label{thm:avgdeg_lowerbound}
Let $G$ be a bipartite graph with $n$ nodes on each side, average degree
$d$, and maximum degree $\Delta$. If $G$ is an $\alpha$-multitasker,
then $\alpha < O(\Delta^{\frac13}/d^{\frac23})$.
\end{theorem}

Note that in case $d=\Omega(\Delta)$ we get $\alpha=O({1/d}^{1/3})$.

\begin{proof}[Proof of~\cref{thm:avgdeg_lowerbound}]
Denote $q:=\lceil \alpha nd/(2\Delta) \rceil$. We use the following
lemma to lower-bound the number of matchings of size $q$ in $G$.
\begin{lemma}\label{lmm:avgdeg_matching_lowerbound}
The number of matchings of size $q$ in $G$
is at least $((1-\alpha)nd)^{q}/q!$.
\end{lemma}
\begin{proof}
Consider the following greedy procedure:
Initialize $G_1\leftarrow G$ and $M\leftarrow\emptyset$.
For $i=1,\ldots,q$, Choose an arbitrary edge $e_i$ in $G_i$, and let $D_i$
denote the set of all edges in $G_i$ sharing an endpoint with $e_i$. Set
$M\leftarrow M\cup\{e_i\}$ and let $G_{i+1}$ be the graph resulting from
removing the edges $\{e_i\}\cup D_i$ from $G_i$.

Initially $G$ has $nd$ edges, and since the maximum degree is
$\Delta$ each iteration removes at most $2\Delta-1$ edges. Hence
for every $i=0,1,\ldots,q$, the number of edges in $G_i$ is at least
$nd-i(2\Delta-1) \geq nd-2q\Delta \geq (1-\alpha)nd$, where the last
inequality is by recalling the setting of $q$. Hence the number of
different matchings that can be realized by the algorithm above is at
least $((1-\alpha)nd)^{q}/q!$.  \end{proof}

We proceed to proving~\cref{thm:avgdeg_lowerbound}.
Consider $H(G,\alpha,q)$.
Let $r\in R$ be an induced matching of size $\alpha q$ in $G$.
Let $G^{-r}$ be the graph resulting from removing all nodes participating
in $r$, together with their incident edges, from $G$.
Note that we remove every edge that has at least one endpoint matched in $r$,
even if its other endpoint does not participate in $r$.
The degree of $r$ in $H$ equals the number of $(1-\alpha)q$-matchings in $G^{-r}$,
which is clearly upper bounded by ${nd\choose (1-\alpha)q}$, since
$G$ has $nd$ edges and $G^{-r}$ is a subgraph of $G$. Furthermore
we clearly have $|R|\leq{n\choose \alpha q}^2$, which implies
that $H$ has in total at most ${n\choose \alpha q}^2{nd\choose
(1-\alpha)q}$ edges. Combining this with the lower bound on $|L|$ given
by~\cref{lmm:avgdeg_matching_lowerbound}, we get the following upper
bound on the average degree of side $L$ in $H$: \[
  \frac{{n\choose \alpha q}^2{nd\choose
  (1-\alpha)q}}{\frac{((1-\alpha)nd)^{q}}{q!}}
  \leq \frac{(\frac{en}{\alpha q})^{2\alpha
  q}(\frac{end}{(1-\alpha)q})^{(1-\alpha)q}}{\frac{1}{\sqrt{2\pi
  q}}\cdot\left(\frac{(1-\alpha)end}{q}\right)^{q}} =
  \sqrt{2\pi q} \left(\frac{1}{1-\alpha}\right)^{(2-\alpha)q}
  \left(\frac{en}{\alpha^2dq}\right)^{\alpha q} \leq \sqrt{2\pi q}
  \left(\frac{4e^4\Delta}{\alpha^3d^2}\right)^{\alpha q},
\]
Where the final inequality is since $1+x\leq e^x$ for every $x$,
and in particular
$\left(\frac{1}{1-\alpha}\right)^{2-\alpha} =
\left(1+\frac{\alpha}{1-\alpha}\right)^{2-\alpha} \leq
(e^{\frac{\alpha}{1-\alpha}})^{2-\alpha} \leq
e^{3\alpha}$ (for $\alpha<1/2$),
and $q > \alpha nd/(4\Delta)$.

If the average degree on side $L$ is less than $1$ then there is an
isolated node in $L$, which represents a $q$-matching in $G$ that contains
no induced matching of size $\alpha q$, which contradicts $G$ being an
$\alpha$-multitasker. Suppose $\alpha > C\cdot\Delta^{1/3}/d^{2/3}$
for a sufficiently large constant $C$. Then the term
$\frac{4e^2\Delta}{\alpha^3d^2}$ is less than $\frac12$. Furthermore
the term $(2\pi q)^{1/(2\alpha q)}$ is less than $2$ as long
as $\alpha \gg \sqrt{\frac{\Delta}{nd}\log(\frac{nd}{\Delta}})$,
which holds for our setting since $\alpha \geq
\frac{\Delta^{1/3}}{d^{2/3}} \geq \left(\frac{\Delta}{nd}\right)^{1/3}
\gg \sqrt{\frac{\Delta}{nd}\log\left(\frac{nd}{\Delta}\right)}$. Hence
$L$ has average degree smaller than $1$ and the proof is finished.
\end{proof}

Note that Theorem \ref{thm:avgdeg_lowerbound} does not provide any
nontrivial bound for $\alpha$ when $\Delta$ exceeds $d^2$. It is,
however, possible to establish nearly the same upper bound provided
by this theorem  with no assumption on $\Delta$. To do so we need
the following lemma, which is proved following the approach of Pyber~\cite{Py}.

\begin{lemma}
\label{l42}
Every (bipartite) graph with $2n$ vertices
and average degree at least $d > 4 \log n$  contains
a subgraph in which the average degree is at least $b=\frac{d}{4 \log n}$
and the maximum degree is at most $2b$.
\end{lemma}
The word bipartite appears in brackets here since any graph $G$
contains a spanning bipartite subgraph  in which the average degree is at
least half of that of $G$, hence the assertion of the lemma holds
for general graphs as well, up to a factor of $2$ in the bound for
$b$.
\begin{proof}
Let $G$ be a bipartite graph with  average degree $d$. As long as
it contains a vertex of degree smaller than $d/2$ omit it.
This process must terminate with a nonempty graph, as the total
number of edges deleted during the process is smaller than $2nd/2$,
that is, smaller than the number of edges of $G$. Thus $G$ contains
a bipartite subgraph $G'$ with minimum degree at least $d/2$. Let
$A$ and $B$ be its vertex classes, where $|A| \geq |B|$.
Let
$A_1 \subset A$ be a minimal nonempty subset of $A$ (with respect to
containment) so that $|N(A_1)| \leq |A_1|.$  There is such a set,
since $|N(A)| =|B| \leq |A|$ and it contains at least $d/2$
vertices as the number of neighbors of any nonempty set is at least
$d/2$. By the minimality $|N(A_1)|=|A_1|$ since otherwise we can
delete a vertex form $A_1$ and get a smaller set satisfying the
condition. It is also clear, by minimality, that $A_1$ satisfies
Hall's condition and thus there is a matching $M_1$ saturating
$A_1$ and $N(A_1)=B_1$. Let $G_1$ be the graph obtained from
$G'$ by removing all vertices besides those in $A_1 \cup B_1 $ and
by removing the perfect matching $M_1$ from it. Then the degree of
every vertex of $A_1$ in $G_1$ is at least $d/2-1$. Let $A_2
\subset A_1$ be a minimal nonempty subset of $A_1$ satisfying
$|N_{G_1}(A_2)| \leq |A_2|$. As before, it clear that $A_2$ exists
(and contains at least $d/2-1$ elements). It is also clear as
before that $A_2$ satisfies Hall's condition and hence there is a
matching $M_2$ saturating $A_2$ and $N(A_2)=B_2$. Proceeding in
this way we get a sequence of $d/2$ matchings $M_1, M_2, M_3,
\ldots , M_{d/2}$ in $G'$ (and hence in $G$), where $M_i$ matches
the vertices of $A_i \subset A$ with those of $B_i \subset B$, and
where $A_{d/2} \subset A_{d/2-1} \subset \cdots \subset A_1$ and
$B_{d/2} \subset B_{d/2-1} \subset \cdots \subset B_1$.
Clearly  $|A_1|=|B_1 | \leq n$ and hence
$n \geq |M_1| \geq |M_2| \geq \cdots \geq |M_{d/2}| \geq 1$.
Thus there is some $i$ so that
$|M_{i+(d/2\log n)-1}| \geq |M_i|/2$. Fix such $i$ and let $H$ be
the union of the matchings $M_i,M_{i+1}, \ldots ,M_{i+(d/2\log
n)-1}$. Define $|M_i|=m$, $2b=\frac{d}{2\log n}$. Then the maximum degree
of $H$ is clearly at most $2b$, as it is the union of $2b$
matchings. The number of vertices of $H$ is $2m$ and its number of
edges is at least $(2b)|M_{i+(d/2\log n)-1}| \geq 2b(m/2)=bm$. Thus
the average degree  of $H$ is at least $b$, completing the proof.
\end{proof}
\begin{theorem}
\label{t43}
Let $G$ be a bipartite graph with $n$ vertices on each side,
and average degree
$d$. If $G$ is an $\alpha$-multitasker,
then $\alpha < O((\frac{\log n}{d})^{1/3}).$
\end{theorem}
\begin{proof}
By Lemma \ref{l42} $G$ contains a subgraph with average degree
$b \geq d/(4 \log n)$ and maximum degree at most $2b$. The result
thus follows from Theorem \ref{thm:avgdeg_lowerbound}.
\end{proof}


A similar reasoning  gives  the following.
\begin{theorem}
\label{t44}
Let $G$ be a bipartite graph with $n$ vertices on each side,
and average degree $d=\Omega(n)$. If $G$ is an $\alpha$-multitasker,
then $\alpha < O((\frac{1}{n})^{1/2}).$
\end{theorem}
\begin{proof}
As proved in~\cite{PRS} using the regularity lemma of Szemer\'edi,
$G$ contains a $d$-regular bipartite graph with $d=\Omega(n)$.
The result thus follows from our upper bound for regular graphs as stated in \Cref{thm:d reg upper bound}.
\end{proof}

\section{Constructions of Good Multitaskers}

It is easy to design arbitrarily large $1$-regular $1$-multitaskers by simply taking disjoint edges,
and $2$-regular $0.5$-multitaskers by taking a cycle of length $n \equiv 0 \pmod 4$).
More generally, one can obtain a $d$ -regular $1/d$-multitaskers by taking $n/d$ disjoint copies of the bipartite clique $K_{d,d}$.
In fact, it is easy to see that any $d$-regular graph is a $1/2d$-multitasker using the greedy algorithm
that given a matching takes in each step an edge in the matching and removes at most $2d-1$ edges that are in conflict with it,
and repeats as long as possible.
The challenge is to design multitaskers achieving $\alpha>0$ that is an absolute constant (independent of $d,k$, and $n$), where both $k$ and $d$ are as large as possible.

\subsection{Several simple constructions}
How can we lower bound the multitasking capability of a network?
It turns out that a simple idea is to contract edges in a given matching and
look for large independent sets in the resulting contracted graph.
We first exemplify this idea when $G$ is a forest.

\begin{lemma}\label{lem:forest}
Let $G$ be a forest. Then $G$ is a $1/2$-multitasker. In other words, if $M$ is a matching in $G$, then $M$ contains an induced matching $M'$ of size at least $|M|/2$.
\end{lemma}

\begin{proof}
Consider an arbitrary matching $e_1,...,e_{|M|}$ in $F$. Contract every edge $e_i \in M$ to a single vertex $v_i$. Since $G$ is a forest, the resulting graph induced on the contracted edges $G[v_1,...,v_{|M|}]$ is a forest, hence it contains an independent set $I$ of size $|M|/2$. The edges corresponding to the vertices in $I$ form an induced matching contained in $M$ of size at least $|M|/2$.
\end{proof}

\begin{remark}\label{rem:alpha_half}
Note that $\alpha \leq 1/2$ holds for any graph which contains a path of length $3$,
as it contains a matching of size $2$ whose largest induced matching has size $1$.
\end{remark}

\begin{remark}
A similar argument also extends to the case where one is concerned with collections of disjoint induced $r$-paths instead of matchings. One simply contracts paths instead of edges of the matcing in the proof of Lemma~\ref{lem:forest}. It is also not hard to generalize the result above to the weighted case, where the edges in the matching have nonnegative weights. We omit the details.
\end{remark}

The argument above can be generalized to minor-closed graph families. For example, we have the following result:
\begin{lemma}\label{lem:planar}
Every planar bipartite graph is a $1/4$-multitasker.
\end{lemma}
\begin{proof}
The proof is similar to Lemma~\ref{lem:forest}. For a matching $e_1,...,e_{|M|}$, the graph obtained by contracting every matching in $M$ is planar. By the four-color Theorem, it has an independent set of size at least $1/4$, concluding the proof.
\end{proof}

We note that the bound $\alpha \geq 1/4$ is tight for bipartite planar graphs. To see this consider the hypercube $H$ over $8$ vertices. It can be seen that $H$ contains a matching of size $4$ that does not contain any induced matching of size greater than $1$, as is demonstrated in Figure~\ref{fig:hypercube}.

Lemmas~\ref{lem:forest} and~\ref{lem:planar} deal with the setting $k=n$, i.e.~they work for matchings of any size, while posing a strict constant bound on the average degree ($d<2$ in the case of forest, and $d<6$ in the planar case). Next we see how to obtain different trade-offs between $k$ and $d$, while keeping $\alpha$ constant. We start with the optimal $\alpha=1/2$,
and prove that
for $k < \Omega(\log_{d}n)$ there exists a $(k, 1/2)$-multitasker.

\begin{theorem}\label{thm:large girth}
Fix $d \in \N$, and let $n \in \N$ be sufficiently large.
There exists a graph that is $(k, 1/2)$-multitasker for all $k \leq s$,
with $s = \Omega(\log_{d}(n))$.
\end{theorem}
\begin{proof}
It is well known there are (explicit) $n$-vertex $d$-regular bipartite graphs of girth $g=\Omega(\log_{d} n)$. Since any edge set of size $g-1$ is a forest, the statement follows from Lemma~\ref{lem:forest}.
\end{proof}

Next, we show that for small constants $\alpha$, we may achieve a significant increase in $k$
by showing existence of a $(O(n/d^{1+4 \alpha}), \alpha)$-multitaskers for any $0<\alpha<1/5$.
\begin{theorem}\label{thm:G n d over n}
Fix $d \in \N$, let $n \in \N$ be sufficiently large, and suppose $\alpha < 1/5$.
There exists a $(k,\alpha)$-multitasker with $n$ vertices on each size, average degree $d$,
for all $k\leq \Omega(n/d^{1+4\alpha})$.
\end{theorem}
\begin{proof}
It is known (see, e.g.,~\cite{feige2016generalized})
that for sufficiently large $n$, there exist an $n$-vertex graph $G(V,E)$ with average degree $d$
such that every subgraph of $G$ of size $s \leq O(n/d^{1+4 \alpha})$
has average degree at most $\frac{1}{2}(\frac{1}{\alpha}-1)$.
Define a bipartite graph $H=(A \cup B,E_H)$ such that $A$ and $B$ are two copies of $V$,
and for $a \in A$ and $b \in B$ we have $(a,b) \in E_H$ if and only if $(a,b) \in E$.
We get that the average degree of $H$ is $d$, and for any two $A' \subseteq A$ and $B' \subseteq B$
such that $|A'| = |B'| \leq s/2$, the average degree of $H[A' \cup B']$ is at most $\frac{1}{\alpha}-1$.
Consider a matching $M$ of size $s/2$ in $H$. By Lemma~\ref{lem:contraction}, if we contract all edges of the matching, we get a graph of average degree at most $\frac{2}{\alpha}-1$. By Lemma~\ref{lem:turan_independent_set}, such a graph contains an independent set of size at least $\frac{1}{2}\alpha|M|$, which corresponds to a large induced matching contain in $M$. This concludes the proof of the theorem.
\end{proof}

\subsection{Regular multitaskers with large  $\alpha(n)$}

The following theorem shows that if we consider only task sets,
i.e.~matchings, of size~\emph{exactly} $n$,
then there are $d$-regular graphs with
$\alpha(n)=\Omega(1/\sqrt{d\log d})$. This nearly matches our upper bound
$O(1/\sqrt d)$ stated in \Cref{thm:d reg upper bound}.

\begin{theorem}\label{thm:d reg alpha = 1/sqrt(d log d)}
There is an absolute constant $c$ such that, for every large enough $d$,
there exists a $d$-regular $G$ such that every perfect matching in $G$
contains an induced matching of size at least $cn/\sqrt{d\log d}$.
\end{theorem}

\begin{proof}
We prove for the setting $d=\frac12n+\sqrt{n \log n}$;
$n$ can then be made larger (while keeping $d$ fixed) by taking disjoint copies. The
construction of $G$ is as follows: Let $A,B$ be the bipartition,
and partition $A$ into $A_1\cup A_2$ and $B$ into $B_1\cup
B_2$, such that $|A_1|=|B_1|=\frac12n+\sqrt{n\log n}$ and
$|A_2|=|B_2|=\frac12n-\sqrt{n\log n}$. The bipartite graphs between $A_1$
and $B_2$ and between $A_2$ and $B_1$ are complete, and there are no
edges between $A_2$ and $B_2$. Let $G_1'$ be a random bipartite graph on
$A_1,B_1$ in which each edge is present independently with probability
$32\sqrt{\frac{\log n}{n}}$. By~\cref{lem:kfactor}
$G_1'$ contains a $(2\sqrt{n\log n})$-regular spanning subgrpah $G_1$,
which we add to $G$. This completes the random construction of $G$,
which is clearly $(\frac12n+\sqrt{n \log n})$-regular.

Next, we argue that with high probability, each subgraph of $G_1$
with $2\sqrt{n\log n}$ nodes on each side has average degree at most
$32\log n$. Clearly it suffices to prove this for $G_1'$. Indeed,
for such a given subgraph of $G_1'$, the expected number of edges is
$32\sqrt{n}\log^{1.5}n$ and hence by the Chernoff bound (\cref{lmm:chernoff}), the probability
to exceed $64\sqrt{n}\log^{1.5}n$ edges (or equivalently average degree
$32\log n$) is at most $\exp(-\frac{32\sqrt{n}\log^{1.5}(n)}{3})) < n^{-10\sqrt{n\log n}}$. There are at
most ${\frac12n+\sqrt{n\log n} \choose 2\sqrt{n\log n}}^2 \leq (\frac{en}{2\sqrt{n\log n}})^{4\sqrt{n\log n}} < n^{2\sqrt{n\log n}}ֿ$,
so by a union bound, the desired property holds with probability $1-o(1)$.

Assume henceforth this event occurs, that is, each subgraph of $G_1$
with $2\sqrt{n\log n}$ nodes on each side has average degree at most
$32\log n$. Consider a perfect matching $M$ in $G$. It
must intersect $G_1$ on at least $2\sqrt{n\log n}$ edges. Let $H$ be
the auxiliary graph whose nodes are these edges of $M\cap G_1$, and two
nodes are neighbors if the corresponding edges of $M$ are connected
in $G_1$. By the above property of $G_1$, $H$ has average degree at
most $400\log n$ and hence by \cref{lem:turan_independent_set}
it contains an independent set of size at least
$\frac{1}{200}\sqrt{n/\log n}$, which correspond to an induced matching
of this size, contained in $M$.  \end{proof}

\begin{lemma}\label{lem:kfactor}
Let $G(V_1,V_2;E)$ be a random bipartite graph with $|V_1|=|V_2|=n$,
in which each edge is present with independently probability
$p=32\sqrt{\frac{\log n}{n}}$. Then, with high probability,
$G$ contains a $(2\sqrt{n\log n})$-regular spanning subgraph.
\end{lemma}
\begin{proof}
By the well known  criterion for containing a factor (see, e.g.,
\cite{LP}, Theorem 2.4.2), $G$ contains a
subgraph as required iff for every $X\subset V_1$ and $Y\subset V_1$,
\begin{equation}
\label{eq:kfactor}
|X| + |Y| + \frac{e(\bar X,\bar Y)}{2\sqrt{n\log n}} \geq n,
\end{equation}
where $e(\bar X, \bar Y)$ denotes the number of edges between $\bar
X=V_1\setminus X$ and $\bar Y=V_2\setminus Y$.
We can restrict attention to $X,Y$ such that $|X|+|Y|\leq n$, as
otherwise~\cref{eq:kfactor} holds trivially. Observe that
\begin{equation}\label{eq:kfactor_expectation}
\E[e(\bar X,\bar Y)]=32\sqrt\frac{\log n}{n}(n-|X|)(n-|Y|)
\end{equation}
and that if we plug
\begin{equation}\label{eq:kfactor_deviation}
e(\bar X,\bar Y)\geq\frac{1}{2}\E[e(\bar X,\bar Y)]
\end{equation}
in the LHS of~\cref{eq:kfactor} then the desired inequality holds,
since by~\cref{eq:kfactor_expectation},
\[
  |X| + |Y| + \frac{e(\bar X,\bar Y)}{2\sqrt{n\log n}} \geq
  |X| + |Y| + \frac{8(n-|X|)(n-|Y|)}{n} =
  8n -7(|X|+|Y|) + \frac{|X||Y|}{n} \geq
  n,
\]
having used $|X|+|Y|\leq n$. Hence
it suffices to show that~\cref{eq:kfactor_deviation} occurs for all
$X,Y$ with high probability.
Assume w.l.o.g.~$|X|\leq|Y|$, which implies $|X|\leq\frac{1}{2}n$. We
consider two cases:
\begin{itemize}
\item $|Y|\leq n-\sqrt{\frac{n}{\log n}}$.
Then we have $(n-|X|)(n-|Y|)\geq\frac12n\sqrt{\frac{n}{\log n}}$,
hence by~\cref{eq:kfactor_expectation} $\E[e(\bar X,\bar Y)]\geq16n$,
and by the Chernoff bound (\cref{lmm:chernoff}),
  \[
    \Pr\left[e(\bar X,\bar Y)\geq\frac{1}{2}\E[e(\bar X,\bar Y)]\right] \geq
   1-\expֿ\left(-\frac18\E[e(\bar X,\bar Y)]\right) \geq
    1-\exp\left(-2n\right).
 \]
  Taking a union bound over at most $4^n$ choices of $X,Y$,
\cref{eq:kfactor_deviation} holds for
all such $X,Y$ with probability $1-o(1)$.
  \item $|Y| > n-\sqrt{\frac{n}{\log n}}$.
  Our assumption $|X|+|Y|<n$ implies in particular that $|Y|\leq n-1$,
  and together with $|X|\leq\frac12n$ we get $(n-|X|)(n-|Y|)\geq\frac12n$.
  Hence by~\cref{eq:kfactor_expectation} $\E[e(\bar X,\bar Y)]\geq16\sqrt{n\log n}$,
 and by the Chernoff bound (\cref{lmm:chernoff}),
  \[
   \Pr\left[e(\bar X,\bar Y)\geq\frac{1}{2}\E[e(\bar X,\bar Y)]\right] \geq
    1-\expֿ\left(-\frac18\E[e(\bar X,\bar Y)]\right) \geq
    1-\exp\left(-2\sqrt{n\log n}\right).
  \]

Noting that the current case assumption together with $|X|+|Y| < n$
implies $|X|<\sqrt{\frac{n}{\log n}}$, we have at most ${n \choose
\sqrt{n/\log n}}^2 \leq n^{2\sqrt{n/\log n}} = 4^{\sqrt{n\log
n}}$ choices for $X,Y$. Taking a union bound over these,
\cref{eq:kfactor_deviation} holds for all such $X,Y$ with probability
$1-o(1)$.
\end{itemize}
A final union bound over the two cases implies that~\cref{eq:kfactor}
holds for all $X,Y$ simultaneously with probability $1-o(1)$.
\end{proof}

\subsection{Construction of $d$-regular multitaskers based on expanders}
In this section we show how to construct multitaskers with multitasking capacity $\Theta_d(\log d/d)$.
This is done based on construction of bipartite spectral expanders. Namely, we have the following result:
\begin{theorem}\label{thm:expander}
    Fix $d \in \N$, and let $n \in \N$ be sufficiently large.
    There exists a $d$-regular bipartite graph with $n$ vertices on each side,
    with $\alpha > \Omega(\frac{\log d}{d}).$
\end{theorem}

We will prove the theorem by showing that if $d$ is large enough constant, and $G$ is a $(n,d,\lambda)$-expander with
$\lambda \leq O(d^{0.9})$, then every matching $M$ in $G$ contains an induced matching of size at least $\frac{|M| \log d}{16d}$.
The proof is similar to a result due to Alon, Krivelevich and Sudakov regarding large independent sets in subgraphs of psuedo-random graphs \cite{alon1999list}.
Given a bipartite $d$-regular graph $G=(A \cup B,E)$ with $|A|=|B|=n/2$,
let $\lambda_1 \geq \lambda_2 \geq...\geq \lambda_n$ be the $n$ eigenvalues of the adjacency matrix of $G$.
It is known that $|\lambda_1|=|\lambda_n|=d$. We let $\lambda$ denote the largest eigenvalue (in absolute value) excluding $\lambda_1, \lambda_n$.
Such $G$ is called a $(n,d,\lambda)$-expander.
We use the following variation of the expander mixing lemma for bipartite $d$-regular graphs:

\begin{lemma}\label{lem:mixing}
Given a bipartite $d$-regular graph $G=(A,B,E)$ with $|A|=|B|=n/2$ we have for every $S\subseteq A$ and $T\subseteq B$,
$$|e(S,T)-\frac{|S||T|d}{n/2}| \leq \lambda \sqrt{|S||T|}.$$
\end{lemma}

Using Lemma ~\ref{lem:mixing} we have the following result:
\begin{lemma}
Let $A'\subseteq A,B' \subseteq B$ with $|A'|=|B'|=an$. Then $$\frac{|e(A',B')|}{2an} \le ad + \lambda/2.$$ In particular, the average degree of $G(A',B')$ is at most $2ad+\lambda.$
\end{lemma}

We first need the following Lemma.
\begin{lemma}\label{lem:iteration}
Let $A'\subseteq A,B' \subseteq B$ with $|A'|=|B'|=an$. Suppose $G(A',B')$ contains a perfect matching.
If $\lambda=o(d)$, the $G(A',B')$ contains an induced matching of size at least
\begin{equation*}
  \frac{n}{4d}\ln\left(1+\frac{md}{n(\lambda/2+1/4)}\right) .
\end{equation*}
\end{lemma}
\begin{proof}
Set $m=an$ and let $M=\{e_1,...,e_m\}$ be perfect matching. Contract all edges in $M$ and call the resulting graph $G_1$, The average degree of $G_1$ is at most $4ad +2 \lambda$.
Pick a vertex $v$ of minimal degree in $G_1$, add it to a set $I$ (initialized to be the empty set) and repeat the process for $G_2=G_1 \setminus \{v  \cup N(v)\}$, where $N(v)$ is the set of neighbors of $v$. Continue iteratively with the above algorithm, until no vertices are left. The crucial observation is that for any $i$ for which $G_i$ is nonempty, if $G_i$ contains $b n$ vertices, then it has average degree at most $4bd + 2 \lambda$. Consider the sequence defined by the recurrence relation
$$a_0=m, a_{i+1}=a_i-\left(4d \frac{a_i}{n}+2 \lambda +1\right)=\left(1-\frac{4d}{n}\right)a_i-(2\lambda+1), \forall i \geq 0.$$
By the definition of our iterative procedure, the cardinality of the graph remaining after $i$ iterations is at least $a_i$. Solving the recurrence above we get that,
$$a_i=\left(1-\frac{4d}{n}\right)^i\left(m+\frac{n(\lambda/2+1/4)}{d}\right)-\frac{n(\lambda/2+1/4)}{d}.$$
It follows that $$a_i \geq e^{-(4d/n)i}\left(m+\frac{n(\lambda/2+1/4)}{d}\right)-\frac{n(\lambda/2+1/4)}{d}.$$
The size of $|I|$ is larger than the smallest index $i$ for which $a_i \leq 0$. Therefore $$|I| \geq \frac{n}{4d}\ln\left(1+\frac{md}{n(\lambda/2+1/4)}\right).$$
The set of edges that corresponds to vertices in $I$ is an induced matching. This concludes the proof.
\end{proof}

Observe that Lemma~\ref{lem:iteration} implies that every $d$-regular bipartite graph with $\lambda \ll d$ and two equal sides contains an induced matching of size $O(\frac{n\log d}{d}).$ We are not aware of a previous proof of this fact.
We can now prove Theorem~\ref{thm:expander}:

\begin{proof}[Proof of Theorem~\ref{thm:expander}]
Suppose first that $|M| \geq n/\ln (d)$. By Lemma~\ref{lem:iteration} $M$ contains an induced matching of size at least
$$\frac{n}{4d}\ln\left(1+\frac{|M|d}{n(\lambda/2+1/4)}\right),$$ which is at least
$$\frac{n}{4d}\ln\left(1+\frac{d}{((\ln d)\lambda/2+1/4)}\right).$$ By our assumptions on $d, \lambda$, we get that
$|M|$ contains an induced matching of size at least $$\frac{n}{8d}\ln\left(1+\frac{d}{\lambda/2+1/4}\right).$$
On the other hand, if $|M| \leq n/\ln (d)$, then the graph induced on $M$, $G[M]$ has average degree at most $2d/\log d+\lambda$.
Therefore, by Lemma~\ref{lem:contraction}, if we contract all edges in $M$ the resulting graph which we denote by $G'[M]$ has average degree at most $4d/ \log d + 2\lambda$. Therefore, $G'[M]$ contains an independent set of size at least $\ell:=\frac{|M|}{4d/ \log d + 2\lambda+1}$. As we assume $\lambda=o(d^{0.9}),$ we have that in this case $M$ contains an induced matching of size at least $\ell$. It is easy to verify that in both cases, we get that $M$ contains an induced matching of size at least $\frac{|M|\log d}{16d}$, concluding the proof.
\end{proof}

Remark: an alternative way to establish the existence of $d$-regular multitaskers with $\alpha(n)=\Omega(\log d/d)$ is take any bipartite $d$-regular graph $G$ of girth at least 7 (e.g., graphs avoiding cycles of length smaller than $7$). Given an arbitrary matching $M$ in $G$, contracting the edges of $M$ results with triangle free graph. As such graphs are known to have an independent set of size $\Omega(\frac{\log d}{d}n)$ it immediately follows that $\alpha=\Omega(\frac{\log d}{d}).$ The drawback of this construction compared to our construction is that $d$ must be sublinear in $n$ in graphs of girth $7$, whereas in the expander based construction, $d$ can be of order $n$. The advantage of the girth construction is that it readily generalizes to depth $r$ networks by simply taking an $r$-partite graph of girth at least $3r+1$. This implies that we can have $\alpha=\Omega(\frac{\log d}{d}n)$ also in networks of depth $r$ so long as we are willing to have $d$ that is sublinear in $n$.

\subsection{The irregular case}

We complement the results above
by providing graphs with average degree $\log \log n$ that are $\alpha$-multitasker
for $\alpha>0$ being a constant independent of $n$.
We start with the following somewhat surprising lower bound.

\begin{theorem}\label{t51}
There exists a bipartite graph $G$ with $n$ vertices in each vertex
class and average degree  at least $\frac{1}{8} \log \log n$
which is a $\frac{1}{20}$-multitasker.
That is, for any integer $k=1,\dots, n$, any matching of size $k$ in $G$ contains an induced
matching of size at least $k/20$.
\end{theorem}
The constant $1/20$ above can be improved, as we show
in \Cref{t52}. We first present a short proof without trying
to optimize the constants.
Note that in view of Theorem
\ref{t43} if the average degree is significantly bigger than $\log n$
then the graph cannot be an $\Omega(1)$ multitasker. It will be
interesting to decide whether or not the $\Omega(\log \log n)$ lower bound
for the
average degree above can be improved to $\Theta(\log n)$.
\begin{proof}
We use the following result, proved in \cite{Al}, Theorem 2.1,
using the method of \cite{PRS}:
\vspace{0.1cm}

\noindent
For every positive integer $M$ and all sufficiently large
$n>n_0(M)$
there exists a bipartite graph $G$
with vertex classes $A$ and $B$, satisfying the following
properties.
\vspace{0.1cm}

\noindent
(i) $|B| \leq |A| =n$.
\vspace{0.1cm}

\noindent
(ii) Every vertex of $A$ has degree $M$ and every vertex of $B$ has
degree larger than $1000 M$.
\vspace{0.1cm}

\noindent
(iii) Every subgraph of $G$ with average degree at least $10$
contains
a vertex of degree at least $1000M$.
\vspace{0.2cm}

\noindent
By examining the proof in \cite{Al} it is not difficult to check
that it works for $M=\frac{1}{8} \log_2 \log_2 n$.
\vspace{0.2cm}

\noindent
Let $G$ be the above graph, with $M=\frac{1}{8} \log_2 \log_2 n$,
after adding to $B$ isolated vertices to make its cardinality
equals that of $A$. Consider now an arbitrary matching
in $G$, and let $a_1b_1, a_2b_2, \ldots ,a_kb_k$ be its edges,
where $a_i \in A$ and $b_i \in B$. Let $H$ be the induced subgraph
of $G$ on the $2k$ vertices $a_i,b_j$. Note that every vertex
$a_i$ has degree at most $M$ in $H$, as this is its degree in $G$,
by property (ii) of $G$. Thus $H$ has at most $Mk$ edges.
As long as the average degree in $H$ is at least $10$, it contains
a vertex $b_i$ of degree at least $1000M$, by property (iii).
In this case we omit $b_i$ and the vertex matched to it $a_i$.
Note that this process cannot omit more than $k/1000 $ pairs of
vertices, as
the total number of edges in $H$ is at most $Mk$. Thus this process
terminates with a matching of size at least $0.999k$ so that the
average degree of its vertices is at most $10$. Consider the graph
$F$ whose vertices are the edges of this matching, where two are
adjacent iff there is an edge connecting them. Then the average
degree in this graph is at most $18$, and hence it contains an
independent set of size at least $0.999k/19 >k/20$. This gives an
induced matching of the required size, completing
the proof.
\end{proof}

Next we show that the constant $1/20$ above can be improved to
nearly $1/3$.
\begin{theorem}\label{t52}
For any fixed small $\epsilon>0$ and large $n$ there
exists a bipartite graph $G'$ with $n$ vertices in each vertex
class and average degree  $(1-o(1))\frac{ \log \log n}{4 \log (10/\epsilon)}$
which is a $(\frac{1}{3}-\epsilon)$-multitasker.
That is, for any integer $s = 1,\dots, n$, any matching of size $s$ in $G$ contains an induced
matching of size at least $(\frac{1}{3}-\epsilon)s$.
\end{theorem}
\begin{proof}
Define $T=\frac{10}{\epsilon}$, $t=\frac{1}{4} \log_T \log_2 n$. Let
$A$ be a set of $n$ vertices, and for each $1 \leq i \leq t$, let
$B_i$ be a set of
$$
\frac{n}{2^{\sqrt {\log n} T^i}}
$$
vertices, where  all sets $A,B_i$ are pairwise disjoint. Let $B$ be
the union of all sets $B_i$ together with $n-\sum_{i=1}^t |B_i|$
additional isolated vertices.  $A$ and $B$ are the two vertex
classes of a bipartite graph
$G$. Each vertex $a \in A$ has one random neighbor in
each set $B_i$ $(1 \leq i \leq t)$, where all choices are uniform
and independent. Thus the degree of every vertex $a \in A$ is
exactly $t$ and hence this is also the average degree of $G$. Our
graph $G'$ will be a spanning subgraph of $G$ obtained by deleting
an edge from each short cycle. We first observe that with high
probability $G$ does not contain too many short cycles.
\vspace{0.2cm}

\noindent
{\bf  Claim 1:}\, With high probability, the number of
cycles of length at most
$10/\epsilon$ in $G$ is $o(n)$.
\vspace{0.2cm}

\noindent
{\bf Proof:}\,
Note, first, that by construction, for every $m$ and every
collection of $m$ potential
edges between the vertices of $G$, the probability that all these
are indeed edges of $G$ is at most
$$
(\frac{1}{|B_t|} )^m =(\frac{1}{n^{1-o(1)}})^m.
$$
(For some such collections of edges, for example ones that contain
at least two neighbors of some $a \in A$ in the same set $B_i$,
the probability is zero, but for any collection the above upper
bound applies). Thus, the probability that there exists
a cycle of length at most $\frac{10}{\epsilon}$ in $G$ is smaller than
$$
\sum_{s=2}^{5/\epsilon} n^{2s} (\frac{1}{n^{1-o(1)}})^{2s}
< 2n^{10 \cdot o(1)/\epsilon} = n^{o(1)}.
$$
The assertion of the claim follows from Markov's Inequality.
\vspace{0.2cm}

\noindent
{\bf Claim 2:}\, The following holds with high probability. For
every $1 <i  \leq t$ and every $s$ satisfying
$s \leq \frac{10}{\epsilon} |B_i|$, the number of edges in any induced
subgraph of $G$ with $s$ vertices in $A$ and $s$ vertices
in $\cup_{j=1}^{i-1} B_j$ is smaller than $(2+\epsilon/4)s$.
\vspace{0.2cm}

\noindent
{\bf Proof:}\,
By the choice of parameters,
$$
\frac{|B_{i-1}|}{|B_i|} =2^{\sqrt {\log n} T^i (1-1/T)}
=(\frac{n}{|B_i|})^{1-\epsilon/10}.
$$
Therefore,
$$
\frac{1}{|B_{i-1}|} =\frac{1}{|B_i|} (\frac{|B_i|}{n})^{1-\epsilon/10}
$$
$$
=(\frac{1}{|B_i|})^{\epsilon/10} (\frac{1}{n})^{1-\epsilon/10}
\leq (\frac{10}{s\epsilon})^{\epsilon/10} (\frac{1}{n})^{1-\epsilon/10}.
$$
Therefore, the probability that there is a subgraph of $G$
with $s$ vertices in $A$, $s$ vertices in $\cup_{j=1}^{i-1} B_j$
and at least $(2+\epsilon/4)s$ edges is at most the following:
$$
{n \choose s}^2 {{s^2} \choose {(2+\epsilon/4)s}}
(\frac{1}{|B_{i-1}|})^{(2+\epsilon/4)s}
\leq
(\frac{en}{s})^{2s} (\frac{es}{2})^{(2+\epsilon/4)s}
(\frac{1}{|B_{i-1}|})^{(2+\epsilon/4)s}
$$
$$
\leq (\frac{en}{s})^{2s} (\frac{es}{2})^{(2+\epsilon/4)s}
[(\frac{10}{s\epsilon})^{\epsilon/10}
(\frac{1}{n})^{1-\epsilon/10}]^{(2+\epsilon/4)s}
$$
$$
\leq
[e^2 (\frac{e}{2})^3 (\frac{10}{\epsilon})^{\epsilon/10(2+\epsilon/4)}
n^{2-(1-\epsilon/10)(2+\epsilon/4)} s^{\epsilon/4 -\epsilon/10(2+\epsilon/4)}]^s
$$
$$
\leq [32(\frac{s}{n})^{\epsilon/20-\epsilon^2/40}]^s
< (\frac{32}{2^{\sqrt {\log n}}})^s,
$$
where here we used the fact that $s \leq \sum_{i=1}^t |B_i|
< \frac{n}{2^{\sqrt
{\log n}}}$.

Summing over  all possible values of $s$ and $i$ we get
$t \cdot O(2^{-\sqrt {\log n}})=o(1)$, completing the proof of the claim.

Fix a  graph $G$ satisfying the assertions of Claims 1 and 2. Let
$G'$ be a graph obtained from $G$ by removing an arbitrary edge
from each cycle of length at most $10/\epsilon$ in $G$. Then $G'$ has
$n$ vertices in each vertex class, and has average degree
$(1-o(1))\frac{\log \log n}{4 \log T}$. To complete the proof we
show that for every $s$, every matching $M$ of size $s$ in $G'$
contains an induced matching (induced in $G'$) of size at least
$(\frac{1}{3}-\epsilon)s$.  We consider two possible cases.
\vspace{0.1cm}

\noindent
{\bf Case 1:}\, $s \leq \frac{10}{\epsilon} |B_t|$.  If
$M$ contains at least $s/3$ edges with endpoints in $B_t$, then
these edges form an induced matching, since every vertex of $A$ has
at most one neighbor in $B_t$ (exactly one neighbor in $G$ and
hence at most one in $G'$). Thus in this case there is an induced
matching of size at least $s/3$. If not, then omit all the edges
containing a vertex in $B_t$. Let $F$ be the following auxiliary
graph. Its vertices are the $g \geq 2s/3$ remaining edges of  the
matching and two are connected if there is an edge of $G'$
connecting the corresponding edges. We have to show that
$F$ contains an independent set on nearly half its vertices.
As Claims 1 and 2 hold, the girth of $F$ is at least
$5/\epsilon$ and for any $p$, any set of $p$ of its vertices spans at
most $(1+\epsilon/4)p$ edges. Construct an independent set in $F$ as
follows. As long as it contains a vertex of degree at most $1$ put
it in the independent set and omit it and its unique neighbor (if
the degree was $1$) from $F$. Suppose that this process stops with
$q$ vertices (hence the independent set so far has at least
$(g-q)/2$ vertices). If $q=0$ we are done, as the independent set
has at least $s/3$ vertices. Otherwise, in the induced subgraph of
$F$ on the remaining $q$ vertices the minimum degree is at least
$2$ and the average degree is at most $2+\epsilon/2$. Hence it contains
at most $\epsilon q/2$ vertices of degree at least $3$. Omit these
vertices. The remaining graph is a union of paths and cycles, which may
contain odd cycles, but all cycles in it are of length at least
$5/\epsilon$. Therefore this part contains an independent set of size
at least $\frac{1}{2}(1-\epsilon/5)(1-\epsilon/2)q$ which together with the
$(g-q)/2$ vertices obtained in the initial process supply an
independent set of size at least
$$
\frac{2s}{3}\frac{1}{2}(1-\epsilon/2)(1-\epsilon/5)
>(1/3-\epsilon)s,
$$
as needed.
\vspace{0.1cm}

\noindent
{\bf Case 2:}\, $s > \frac{10}{\epsilon} |B_t|$.  Note that $s \leq
\sum_{i=1}^t |B_i| =(1+o(1))|B_1|< \frac{10}{\epsilon} |B_1|$.
Choose $i$ so that
$$
\frac{10}{\epsilon} |B_{i+1}|< s \leq \frac{10}{\epsilon} |B_i|.
$$
Thus  $1 \leq i <i+1 \leq t$.
Note, first, that the number
of edges of $M$ containing a vertex from $\cup_{j>i} B_j$ is
at most $\cup_{j>i} |B_j| =(1+o(1))|B_{i+1}| \leq
(1+o(1))\frac{\epsilon}{10} s$. Omit the edges of the matching containing these
vertices and proceed as before. If there are at least, say,
$(1/3-\epsilon)s$ edges of the matching containing a vertex from
$B_i$ (a condition that holds automatically if $i=1$),
these edges form an induced matching and the desired result
follows. Else omit these edges and construct the graph $F$ whose
vertices are the remaining edges of the matching (there are at
least $2s/3$ of them), where two are adjacent iff there is an edge
of $G'$ connecting them. This graph has girth at least $5/\epsilon$ and
for every $p$, any set of $p$ of its vertices spans at most
$(1+\epsilon/4)p$ edges. Thus it contains an independent set on
at least a fraction of $\frac{1}{2}(1-\epsilon/2)(1-\epsilon/5)$ of its vertices,
completing the proof for this case and hence also the proof of
the theorem.
\end{proof}
\vspace{0.2cm}

\noindent
{\bf Remark:}\, The graph $G'$ constructed in the proof of Theorem
\ref{t52} does not have a perfect matching, and in fact
has many isolated vertices in the set $B$. It is easy to modify it
and construct a bipartite graph $G''$ which is a
$(\frac{1}{4}-\epsilon)$-multitasker with average degree
$\Omega(\log \log n)$ and contains a perfect matching. Indeed,
the construction of $G$ implies that with high probability each
vertex in $B'= \cup_{i=1}^t B_i$  has degree (much) bigger than
$t$, which is the degree of each vertex of $A$. Therefore, by
Hall's Theorem, $G$ contains a matching saturating all vertices of
$B'$. When constructing $G'$ from $G$ by omitting an edge from
each short cycle, keep all edges of this matching (by simply
omitting an edge not in this matching from each short cycle).
Now add a perfect matching from the vertices in
$B-B'$ (that is, the isolated vertices in $B$) to the unsaturated
vertices in $A$. The resulting graph, call it $G''$, contains a
perfect matching. In addition, it is a
$(\frac{1}{4}-\epsilon)$-multitasker. To see this note that all newly added
edges form an induced matching in $G''$, as their $B$-vertices are
of degree $1$. Thus if at least $1/4$ of the edges of a given
matching $M$ are among the new edges, we get an induced
matching of size at least $|M|/4$. Otherwise, at least
$3|M|/4$ of the edges of $M$ belong to the graph $G'$, and hence
contain an induced matching of size at least
$(\frac{1}{3}-\epsilon)\frac{3|M|}{4} > (\frac{1}{4}-\epsilon)|M|$.
\vspace{0.2cm}

\noindent
Finally we mention the following result for graphs with average
degree $\Theta(\log n)$.
\begin{proposition}
\label{p53}
There exists an absolute positive constant $c$ and
a bipartite graph $G$ with $n$ vertices in each vertex
class and average degree  at least $2 \log(2n)$
which is a $c/\sqrt{\log k}$ multitasker, that is, for any integer
$k$, any matching of size $k$ in $G$ contains an induced
matching of size at least $c k/\sqrt{\log k}$.
\end{proposition}

The proof is similar to the previous one, using the assertion and
proof of Theorem 2.1 in \cite{Al}. We omit the details.

\section{Conclusions}
The limited ability to perform multiple tasks at the same time is one of the most salient and defining
characteristics of human cognition. Despite this fact, parallel processing capabilities of neural systems
remain largely unexplored.
We have considered a new multitasking measure for parallel architectures that is aimed at providing quantitative measures for such capabilities.
We established an inherent tradeoff between the density of the network and its multitasking capacity that holds for every graph that is sufficiently dense.
This tradeoff is rather general and it applies to regular graphs, to irregular graphs and to layered networks of depth greater than $2$. We have also obtained quantitative
insights. For example, we have shown that our upper bound on multitasking capacity is tight for regular graphs and tasks sets of size $n$, provided evidence that interference increases as depth increases from $2$ to $r>2$ and demonstrated that irregular graphs allow for better multitasking than regular graphs for certain edge densities. Our findings are also of interest to
recent effort in cognitive neuroscience to pinpoint the reason for the stark limitations people experience in multiasking control demanding tasks. While our graph-theoretical model is very far from modeling real biological networks, it appears that establishing multitasking limitations for such simple models is necessary before we can address more complicated settings.

We have also considered network architectures that reduce interference and found that networks with pseudorandom properties (locally sparse, spectral expanders, graphs with high girth)
have good multitasking capabilities. Interestingly, previous works have documented the benefits of random and pseudorandom architectures in deep learning, Hopfield networks and other settings \cite{arora2014provable,valiant2000circuits,komlos1988convergence}. Whether there is an underlying cause for these results remains an interesting direction for future research.

Our work is still limited in several aspects.
First, our model is graph-theoretic in nature, focusing exclusively on the adjacency structure of tasks and does not consider many parameters that emerge in biological and artificial parallel architectures. Second, we do not address tasks of different weights (assuming all tasks have the same weights), stochastic and probabilistic interference (we assume interference occurs with probability 1) and the exact implementation of the functions that compute the tasks represented by edges.In sum, while we hope that we have convinced the reader that our graph theoretic approach already captures interesting issues of multitasking, and entails nontrivial observations, to achieve a greater realism and predictive value, one will need to go beyond the graph theoretic structure and consider other parameters that arise in neural networks.


To summarize, the work we have presented here takes an important step towards laying the foundations for a deeper understanding of the factors that affect the tension between
efficiency of representation, and flexibility of processing in neural network architectures. We hope
that this will help inspire a parallel proliferation of efforts to further explore this area.

\bibliographystyle{alpha}
\bibliography{multi}

\appendix
\section{Appendix: Bounds on the number of $k$-matchings}\label{section:num_k_matching}
In this section we derive Corollary~\ref{eq:num_k_matchings} from Lemma~\ref{lmm:num_k_matchings}, which states that $M_k(G)  \geq {n\choose k}^2\left(1-\frac{k}{nd}\right)^{nd-k} \left(\frac{kd}{n}\right)^k$. We now bound the first two terms from below. For the first term (the binomial coefficient), we have
\[
  {n\choose k} =
  \frac{n(n-1)\ldots(n-k+1)}{k!} \geq
  \frac{(n-k)^k}{k^ke^{-k}\sqrt{2\pi k}} =
\]
\[
  \left(1-\frac{k}{n}\right)^k\cdot\left(\frac{en}{k}\right)^k\cdot\frac{1}{\sqrt{2\pi k}} \geq
  \left(\frac{1}{2e}\right)^{k^2/n}\cdot\left(\frac{en}{k}\right)^k\cdot\frac{1}{\sqrt{2\pi k}},
\]
where we use fact that $\left(1-\frac{k}{n}\right)^k \geq \left(\frac{1}{2e}\right)^{k^2/n}$
for all $k<n/2$. For the second term, we first use the following:
\begin{lemma}\label{lmm:approx1e}
For every $x\geq2$,  $(1-\frac{1}{x})^x \geq \frac{1}{e}-\frac{7}{6ex}$.
\end{lemma}
\begin{proof}
$\sum_{k=3}^\infty\frac{1}{kx^k} \leq \frac{1}{3x^2}\sum_{k=1}^\infty\frac{1}{x^k} \leq \frac{2}{3x^2}$ since $x\geq2$, and hence
$
  \left(1-\frac{1}{x}\right)^x =
  e^{x\ln(1-\frac{1}{x})}  =
  e^{x(-\frac{1}{x}-\frac{1}{2x^2}-\frac{1}{3x^3}\ldots)} \geq
  e^{x(-\frac{1}{x}-\frac{7}{6x^2})} =
  \frac{1}{e}\cdot e^{-\frac{7}{6x}} \geq
  \frac{1}{e}\left(1-\frac{7}{6x}\right)
$.
\end{proof}
Now we may bound,
\begin{align*}
\left(1-\frac{k}{nd}\right)^{nd-k} &\geq \left(1-\frac{k}{nd}\right)^{nd} & \\
&\geq \left(\frac{1}{e}-\frac{7k}{6end}\right)^k & \text{by Lemma~\ref{lmm:approx1e}} \\
&= \frac{1}{e^k}\left(1-\frac{7k}{6nd}\right)^k & \\
&\geq \frac{1}{e^k}\left(\frac{1}{e}-\frac{49k}{36end}\right)^{\frac{7k^2}{6nd}} & \text{by Lemma~\ref{lmm:approx1e}} \\
&\geq \frac{1}{e^k}\left(\frac{1}{2e}\right)^{\frac{7k^2}{6nd}} & ,
\end{align*}
where for the last inequality we use $\frac{49k}{36end} < \frac{1}{2e}$,
which follows from the assumption that $k < n/2 \leq nd/4$ (for $d \geq 2$).
Plugging both bounds into Lemma~\ref{lmm:num_k_matchings}, we get
\begin{eqnarray*}
  M_k(G) & \geq &
    {n\choose k}^2\left(1-\frac{k}{nd}\right)^{nd-k} \left(\frac{kd}{n}\right)^k \\
  & \geq & \left(\frac{1}{2e}\right)^{\frac{2k^2}{n}}\cdot\left(\frac{en}{k}\right)^{2k}\cdot\frac{1}{2\pi k} \times \frac{1}{e^{k}}\left(\frac{1}{2e}\right)^{\frac{7k^2}{6nd}} \cdot \left(\frac{kd}{n}\right)^k \\
  & = &
  \left(\frac{end}{k}\right)^k \cdot \left(\frac{1}{2e}\right)^{\frac{k^2}{n}\left(2+\frac{7}{6d}\right)} \cdot \frac{1}{2\pi k} \geq
  \left(\frac{end}{k}\right)^k \cdot \left(\frac{1}{2e}\right)^{4k^2/n} \cdot \frac{1}{2\pi k}.
\end{eqnarray*}


\end{document}